\documentclass[runningheads,envcountsame,a4paper,12pt]{llncs}
\usepackage[margin=2.3cm,marginpar=1.6cm]{geometry}
\usepackage{graphicx}
\usepackage{latexsym}
\usepackage[utf8]{inputenc}
\usepackage{booktabs}
\usepackage{arydshln}
\usepackage{amssymb}
\usepackage{enumerate}
\usepackage{graphicx} 
\usepackage{amssymb,latexsym,times,amsmath,xspace}
\usepackage{enumerate,verbatim}
\usepackage{rotating}
\usepackage{environ}
\usepackage{multicol}
\usepackage{tikz}

\newcommand{\logicFont}[1]{\protect\ensuremath{\mathrm{#1}}\xspace}
\newcommand{\problemmaFont}[1]{\protect\ensuremath{\mathsf{#1}}\xspace}

\newcommand{\classFont}[1]{\protect\ensuremath{\mathsf{#1}}\xspace}

\newcommand{\ML}{\logicFont{ML}}

\newcommand{\PL}{\logicFont{PL}}
\newcommand{\QPL}{\logicFont{QPL}}

\newcommand{\FOprob}{\logicFont{FO}(\perp\!\!\!\perp_{\rm c}, \approx)}

\newcommand{\esofu}{\logicFont{ESOf}_{[0,1]}}
\newcommand{\esof}{\logicFont{ESOf}_{[0,1]}(\logicFont{SUM},\times)}
\newcommand{\esofsum}{\logicFont{ESOf}_{[0,1]}(\logicFont{SUM})}
\newcommand{\pind}{\perp\!\!\!\perp}
\newcommand{\cpind}{\perp\!\!\!\perp_{\rm c}}
\newcommand{\upind}{{}_u\hspace{-1.5mm}\pind}
\newcommand{\fpind}{{}_\forall\hspace{-1.5mm}\pind}
\newcommand {\pci}[3] {#2~\!\!\pind_{#1}\!\!~#3}
\newcommand {\pmi}[2] {#1~\!\!\pind\!\!~#2}
\newcommand {\upmi}[2] {#1 ~\!\!\upind\!\!~ #2}
\newcommand {\fpmi}[2] {#1 ~\!\!\fpind\!\!~ #2}
\newcommand{\SUM}{\mathrm{SUM}}

\newcommand{\x}{\tuple x}
\newcommand{\y}{\tuple y}
\newcommand{\z}{\tuple z}

\newcommand{\pcixyz}{\pci{\tuple x}{\tuple y}{\tuple z}}
\newcommand{\pmixy}{\pmi{\tuple x}{\tuple y}}
\newcommand{\modelsr}{\models^{\geq 0}}
\newcommand{\modelsd}{\models^{[0,1]}}

\newcommand{\Th}{\problemmaFont{Th}}

\newcommand{\TOWER}{\classFont{TOWER}(poly)}

\newcommand{\PSPACE}{\classFont{PSPACE}}

\newcommand{\EXPSPACE}{\classFont{EXPSPACE}}
\newcommand{\twoEXPSPACE}{\classFont{2\text-EXPSPACE}}

\newcommand{\AEXPTIME}{\classFont{AEXPTIME}}
\newcommand{\AtwoEXPTIME}{\classFont{2\text-AEXPTIME}}

\newcommand{\tuple}[1]{\vec{#1}}

\newcommand{\Dom}{\operatorname{Dom}}
\newcommand{\ran}{\operatorname{Ran}}

\newcommand{\ar}{\operatorname{ar}}

\newcommand{\Fr}{\operatorname{Fr}}
\newcommand{\A}{\mathfrak{A}}

\newcommand{\Z}{\mathbb{Z}}
\newcommand{\Po}{\mathcal{P}}

\newcommand{\calL}{\mathcal{L}}
\newcommand{\calC}{\mathcal{C}}
\newcommand{\M}{\mathcal{M}}
\newcommand{\X}{\mathbb{X}}
\newcommand{\Y}{\mathbb{Y}}

\renewcommand{\a}{\alpha}
\renewcommand{\b}{\beta}

\def\dep{=\!\!}
\newcommand{\sub}{\subseteq}
\newcommand{\approxt}{\approx^*}

\newcommand{\kcup}{\sqcup_{k}}
\newcommand{\qcup}[1]{\sqcup_{#1}}
\newcommand{\Var}{\mathrm{Var}}
\newcommand{\negg}{\sim\hspace{-1mm}}
\newcommand{\Varfo}{\mathrm{Var_{1}}}
\newcommand{\Varso}{\mathrm{Var_{2}}}

\def\dep{=\!\!}
\newcommand{\FO}{{\rm FO}}

\newcommand{\ESO}{{\rm ESO}}

\newcommand{\f}{\phi}

\newcommand{\const}[1]{=\!\!(#1)}

\newcommand{\repeatproposition}[1]{%
  \begingroup
  \renewcommand{\theproposition}{\ref{#1}}%
  \expandafter\expandafter\expandafter\proposition
  \csname repproposition@#1\endcsname
  \endproposition
  \endgroup
}

\NewEnviron{repproposition}[1]{%
  \global\expandafter\xdef\csname repproposition@#1\endcsname{%
    \unexpanded\expandafter{\BODY}%
  }%
  \expandafter\proposition\BODY\unskip\label{#1}\endproposition
}

\newcommand{\repeattheorem}[1]{%
  \begingroup
  \renewcommand{\thetheorem}{\ref{#1}}%
  \expandafter\expandafter\expandafter\theorem
  \csname reptheorem@#1\endcsname
  \endtheorem
  \endgroup
}

\NewEnviron{reptheorem}[1]{%
  \global\expandafter\xdef\csname reptheorem@#1\endcsname{%
    \unexpanded\expandafter{\BODY}%
  }%
  \expandafter\theorem\BODY\unskip\label{#1}\endtheorem
}

\newcommand{\repeatlemma}[1]{%
  \begingroup
  \renewcommand{\thelemma}{\ref{#1}}%
  \expandafter\expandafter\expandafter\lemma
  \csname replemma@#1\endcsname
  \endlemma
  \endgroup
}

\NewEnviron{replemma}[1]{%
  \global\expandafter\xdef\csname replemma@#1\endcsname{%
    \unexpanded\expandafter{\BODY}%
  }%
  \expandafter\lemma\BODY\unskip\label{#1}\endlemma
}

\begin{document}
\title{Facets of Distribution Identities in Probabilistic Team Semantics\thanks{The first and the third author were  supported by grant 308712, the fourth by grant 285203 of the Academy of Finland.}}
%
%
\author{Miika Hannula\inst{1}\orcidID{0000-0002-9637-6664}\and
\AA sa Hirvonen\inst{1}\orcidID{0000-0003-2149-4153} \and
Juha Kontinen\inst{1}\orcidID{0000-0003-0115-5154} \and
Vadim Kulikov\inst{1,2} \and
Jonni Virtema\inst{3}\orcidID{0000-0002-1582-3718}}%
\authorrunning{M. Hannula et al.}
%
\institute{University of Helsinki, Finland \email{\{miika.hannula,asa.hirvonen,juha.kontinen\}@helsinki.fi} \and
Aalto University, Finland,
\email{vadim.kulikov@iki.fi}
 \and
Hasselt University, Belgium,
\email{jonni.virtema@uhasselt.be}
}
\maketitle              
\begin{abstract}
  We study probabilistic team semantics which is a semantical
  framework allowing the study of logical and probabilistic
  dependencies simultaneously. We examine and classify the expressive
  power of logical formalisms arising by different probabilistic atoms
  such as conditional independence and different variants of marginal
  distribution equivalences. We also relate the framework to the
  first-order theory of the reals and apply our methods to the open
  question on the complexity of the implication problem of conditional
  independence.


%

\end{abstract}
%
%
%

\section{Introduction}

Team semantics, introduced by Hodges \cite{hodges97} and popularised by
V\"a\"an\"anen \cite{vaananen07}, shifts the focus of logics away from
assignments as the primitive notion connected to satisfaction. In team
semantics formulae are evaluated with respect to sets of assignments
(i.e., teams)
as opposed to single assignments of Tarskian semantics. During the last decade the research
on team semantics has flourished, many logical
formalisms have been defined, 
and surprising
connections to other fields identified.
In particular, several promising application areas of team semantics have been identified recently.
%
Krebs et al. \cite{kmvz18} developed  a team based approach to linear temporal logic  for the verification of information flow properties.
In applications to database theory, a team corresponds exactly to a database table (see, e.g., \cite{HannulaK16}).
Hannula et al. \cite{HannulaKV18} introduced a framework that extends the connection of team semantics and database theory to polyrelational databases and data exchange. 

The focus of this article is probabilistic team semantics which connects team based logics to probabilistic dependency notions. Probabilistic team semantics is built compositionally upon the notion of a probabilistic team, that is, a probability distribution over variable assignments. While the
first ideas of probabilistic teams trace back to the works of
Galliani \cite{galliani08} and Hyttinen et al. \cite{Hyttinen15b}, the systematic study of the topic was
initiated and further continued by Durand et al. in \cite{DurandHKMV18,HKMV18}. It is worth noting that in \cite{EPTCS286.2}
  so-called causal teams have been introduced to  logically model causality and interventions.
Probabilistic team semantics has also a close connection to the area of  metafinite model theory \cite{GradelG98}. 
In metafinite model theory, finite structures are extended with an another (infinite) domain sort such as the real numbers (often with arithmetic) and with weight functions that work as a bridge between the two sorts. This approach provides an elegant way to model  weighted graphs and other structures that refer to infinite structures. The exact relationship between   probabilistic team semantics and  logics over metafinite models  as well as with probabilistic databases of \cite{Cavallo:1987} will be a topic of future research.

The starting point of this work comes from \cite{HKMV18} in which probabilistic team semantics was defined following the lines of \cite{galliani08}.  The main theme in \cite{HKMV18} was to characterize logical formalisms in this framework in terms of existential second-order logic. Two main probabilistic dependency atoms were examined. The probabilistic conditional independence atom $\pcixyz$ states that the two variable tuples $\tuple y$ and $\tuple z$ are independent given the third tuple $\tuple x$. The marginal identity atom $\tuple x \approx \tuple y$ states that the marginal distributions induced from the two tuples  $\tuple x$ and $\tuple y$ (of the same length) are identical. The extension of first-order logic with these atoms ($\FO(\cpind,\approx)$) was then shown to correspond to a two-sorted variant of existential second-order logic that allows a restricted access to arithmetical operations for numerical function terms. What was left unexamined were the relationships between different logical formalisms in probabilistic team semantics. In fact, it was unknown whether there are any meaningful probabilistic dependency notions such that the properties definable with one notion are comparable to those definable with another. 

In this article we study the relative expressivity of first-order logic with probabilistic conditional independence atoms ($\FO(\cpind)$) and with marginal identity atoms ($\FO(\approx)$). The logic
$\FO(\approx)$ is a probabilistic variant of \emph{inclusion logic} that is strictly less expressive than \emph{independence logic}, after which $\FO(\cpind)$ is modelled \cite{galliani12,gradel10}.
In addition, we examine $\FO(\approx^*)$ which is another extension defined in terms of so-called marginal distribution equivalence. The \emph{marginal distribution equivalence atom} $\tuple x\approx^* \tuple y$ for two variable tuples $\tuple x$ and $\tuple y$ (not necessarily of the same length) relaxes the truth condition of the marginal identity atom in that the two distributions induced from $\tuple x $ and $\tuple y$  are required to determine the same multisets of probabilities. The aforementioned open question is now answered in the positive. The logics mentioned above are not only comparable, but they form a linear expressivity hierarchy: $\FO(\approx) < \FO(\approxt) \leq \FO(\pind_{\rm c})$. 
 We also show that $\FO(\approx)$ enjoys a union closure property that is a generalization of the union closure property of inclusion logic, and that  conditional independence atoms $\pcixyz$ can be defined with an access to only marginal independence atoms $\pmi{\tuple x}{\tuple y}$ between two variable tuples. Furthermore, we show that, surprisingly, $\FO(\approxt)$ corresponds to $\FO(\approx, \dep(\cdot))$, where $\dep(\cdot)$ refers to the dependence atom defined as  a declaration of functional dependence over the support of the probabilistic team. The question whether  $\FO(\approx, \dep(\cdot))$  is strictly less expressive than  $\FO(\cpind)$ is left as an open question; in team semantics the corresponding logics are known to be equivalent. The above findings look outwardly very similar to many results in team semantics.
 However, it is important to note that, apart perhaps from the union closure property, the results of this paper base on entirely new ideas and do not recycle old arguments from the team semantics context.

We also investigate (quantified) propositional logics with probabilistic team semantics. By connecting these logics to the arithmetic of the reals  we show upper bounds for their associated computational problems. Our results suggest that the addition of probabilities to team semantics entails an increase in the complexity. Satisfiability of propositional team logic ($\PL(\sim)$), i.e., propositional logic with classical negation is in team semantics known to be complete for alternating exponential time with polynomially many alternations  \cite{HannulaKVV18}. Shifting to probabilistic team semantics analogous problems are here  shown to enjoy double exponential space upper bound. This is still lower than the complexity of satisfiability for modal team logic ($\ML(\sim)$) in team semantics, known to be complete for the non-elementary complexity class $\TOWER$ which consists of problems solvable in time restricted by some tower of exponentials of polynomial height \cite{Luck18}. One intriguing consequence of our translation to real arithmetic is that the implication problem of conditional independence statements over binary distributions is decidable in exponential space. The decidability of this problem is open relative to  all discrete probability distributions \cite{NiepertGSG13}.




\section{Preliminaries}

First-order variables are denoted by $x,y,z$ and tuples of first-order
variables by $\vec{x},\vec{y},\vec{z}$. By $\Var(\tuple x)$ we denote
the set of variables that appear in the variable sequence $\tuple x$.
The length of the tuple $\vec{x}$ is denoted by
$\lvert
\vec{x}\rvert$. 
A \emph{vocabulary} $\tau$ is a set of relation symbols and function
symbols with prescribed arities. We mostly denote relation symbols by
$R$ and function symbols by $f$, and the related arities by $\ar(R)$
and $\ar(f)$, respectively.
The closed interval of real numbers between $0$ and $1$ is denoted by
$[0,1]$.  Given a finite set $A$, a function $f\colon A\to[0,1]$ is
called a \emph{(probability) distribution} if $\sum_{s\in
  A}f(s)=1$. In addition, the empty function is a \emph{distribution}.

The probabilistic logics investigated in this paper are extensions of
first-order logic $\FO$ over a vocabulary $\tau$ given by the grammar rules:
\[
\phi ::= x=y \mid x\neq y \mid R(\vec{x}) \mid \neg R(\vec{x}) \mid
(\phi\land\phi) \mid (\phi\lor\phi) \mid \exists x\phi \mid \forall x
\phi,
\]
where $\vec{x}$ is a tuple of first-order variables and $R$ a relation
symbol from $\tau$.

Let $D$ be a finite set of first-order variables and $A$ be a nonempty
set. A function $s\colon D \to A$ is called an \emph{assignment}. 
 For a variable $x$ and $a\in A$, the
assignment $s(a/x)\colon D\cup\{x\} \rightarrow A$ is equal to $s$
with the exception that $s(a/x)(x)=a$.
  A $\emph{team}$
$X$ is a finite set of assignments from $D$ to
$A$. The set $D$ is called the \emph{domain} of $X$ (written $\Dom(X)$) and the set $A$ the \emph{range} of $X$ (written $\ran(X)$). Let $X$ be a team with range $A$, and let
 $F\colon X\to \Po(A)\setminus \{\emptyset\}$ be a function.  We
 denote by $X[A/x]$ the modified team
 $\{s(a/x) \mid s\in X, a\in A\}$, and by $X[F/x]$ the team
 $\{s(a/x)\mid s\in X, a\in F(s)\}$. 
%
A $\emph{probabilistic team}$ $\X$ is a distribution
$\X\colon X\rightarrow [0,1]$. 
 Let $\A$ be a $\tau$-structure and $\X:X\to [0,1]$ a probabilistic team such that the domain of
$\A$ is the range of $X$.  Then we say that $\X$ is a
probabilistic team of $\A$.  In the following, we will define two
notations $\X[A/x]$ and $\X[F/x]$. Let $\X\colon X\to[0,1]$ be a
probabilistic team, $A$ a finite non-empty set, $p_A$ the set of all probability
distributions $d\colon A \to [0,1]$, and $F\colon X\to p_A$ a
function.
%
We denote by $\X[A/x]$ the probabilistic team $X[A/x] \to [0,1]$ such
that
\[
\X[A/x](s(a/x)) = \sum_{\substack{t\in X \\ t(a/x) = s(a/x)}} \X(t)
\cdot \frac{1}{\lvert A \rvert},
\]
for each $a\in A$ and $s\in X$. Note that if $x$ does not belong to the domain of $X$
then the righthand side of the above equation is simply
$\X(s) \cdot \frac{1}{\lvert A \rvert}$.  By $\X[F/x]$ we denote the
probabilistic team $X[A/x] \to [0,1]$ defined such that
\[
\X[F/x](s(a/x)) = \sum_{\substack{t\in X \\ t(a/x) = s(a/x)}} \X(t)
\cdot F(t)(a),
\]
for each $a\in A$ and $s\in X$. Again if $x$ does not belong to the domain of $X$,
$\sum$ can be dropped from the above equation.

If $\Y\colon X\to [0,1]$ and $\Z\colon X \to [0,1]$ are probabilistic
teams and $k\in [0,1]$, then we write $\Y\kcup \Z$ for the $k$-scaled union of $\Y$ and
$\Z$, that is, the probabilistic team $\Y\kcup \Z\colon X\to [0,1]$
defined such that $(\Y\kcup \Z)(s)=k\cdot \Y(s) + (1-k)\cdot \Z(s)$
for each $s\in X$.

%
%
%
We may now define probabilistic team semantics for first-order
formulae. The definition is the same as in \cite{HKMV18}. The only
exception is that it is here applied to probabilistic teams that have
real probabilities, whereas in \cite{HKMV18} rational probabilities were used.

\begin{definition}\label{def:probsem}
  Let $\A$ be a probabilistic $\tau$-structure over a finite domain
  $A$, and $\X\colon X \to [0,1]$ a proba\-bi\-listic team of
  $\A$. The satisfaction relation $\models_\X$ for first-order logic
  is defined as follows:
  
  \begin{tabbing}
    left \= $\A \models_{\X} (\psi \land \theta)$ \= $\Leftrightarrow$ \= $\forall s\in X: s(\tuple x) \in R^{\A}$\kill
    \> $\A \models_{\X} x=y$ \> $\Leftrightarrow$ \> $\text{for all } s\in X: \text{if }\X(s)> 0 \text{, then } s(x)=s(y)$\\ 
   \> $\A \models_{\X} x\neq y$ \> $\Leftrightarrow$ \> $\text{for all } s\in X: \text{if }\X(s)> 0 \text{, then } s(x) \not= s(y)$\\
    \> $\A \models_{\X} R(\tuple x)$ \> $\Leftrightarrow$ \> $\text{for all } s\in X: \text{if }\X(s)> 0 \text{, then } s(\tuple x) \in R^{\A}$\\ 
    \> $\A \models_{\X} \neg R(\tuple x)$ \> $\Leftrightarrow$ \> $\text{for all } s\in X: \text{if } \X(s)> 0 \text{, then } s(\tuple x) \not\in R^{\A}$\\
    \> $\A\models_{\X} (\psi \land \theta)$ \> $\Leftrightarrow$ \> $\A \models_{\X} \psi \text{ and } \A \models_{\X} \theta$\\
    \> $\A\models_{\X} (\psi \lor \theta)$ \> $\Leftrightarrow$ \> $\A\models_\Y \psi \text{ and } \A \models_\Z \theta \text{ for some  $\Y,\Z,k$}$ s.t.\ $\Y\kcup \Z = \X$\\
    \> $\A\models_{\X} \forall x\psi$ \> $\Leftrightarrow$ \> $\A\models_{\X[A/x]} \psi$\\
    \> $\A\models_{\X} \exists x\psi$ \> $\Leftrightarrow$ \> $\A\models_{\X[F/x]} \psi \text{ holds for some } F\colon X \to p_A $.
  \end{tabbing}
\end{definition}
Probabilistic team semantics is in line with Tarskian semantics for first-order formulae ($\models_s$):
\[
\A \models_{\X} \psi  \Leftrightarrow  \forall s\in X \text{ such that } \X(s)>0 : \A \models_s\psi. 
\]
In particular the non-classical semantics for negation is required for the above equivalence to hold.

In this paper we consider three probabilistic atoms: marginal
identity, probabilistic independence, and marginal distribution
equivalence atom. The first two were first introduced in the context of
multiteam semantics in \cite{DurandHKMV18}, and they extend the
notions of inclusion and independence atoms from team semantics
\cite{galliani12}.

We define $|\X_{\tuple x =\tuple a}|$ where $\tuple x$ is a tuple of
variables and $\tuple a$ a tuple of values, as 
\[|\X_{\tuple x =\tuple a}|:=\sum_{\substack{s(\tuple x)=\tuple
    a\\s\in X}} \X(s).\]
If $\phi$ is some first-order formula, then $|\X_{\phi}|$ is defined
analogously as the total sum of weights of those assignments in $X$
that satisfy $\phi$.

If $\tuple x,\tuple y$ are variable sequences of length $k$, then
$\tuple x \approx \tuple y$ is a \emph{marginal identity atom} with
the following semantics:
\begin{equation}\label{eq:marg}
  \A \models_{\X} \vec{x} \approx \vec{y} \Leftrightarrow \lvert {\X}_{\vec{x}=\vec{a}} \rvert =  \lvert {\X}_{\vec{y}=\vec{a}} \rvert \text{ for each $\vec{a}\in A^k$}.
\end{equation}
Note that the equality
$\vert{\X}_{\vec{x}=\vec{a}} \rvert = \lvert {\X}_{\vec{y}=\vec{a}}
\rvert$
in \eqref{eq:marg} can be equivalently replaced with
$\lvert{\X}_{\vec{x}=\vec{a}} \rvert \leq \lvert
{\X}_{\vec{y}=\vec{a}}\rvert$
since the tuples $\tuple a$ range over $A^k$ for a finite $A$ (see \cite[Definition 7]{DurandHKMV18} for details). Due to this alternative
formulation, marginal identity atoms were in
\cite{DurandHKMV18} called probabilistic
inclusion~atoms. 
Intuitively, the atom $\tuple x \approx \tuple y$
states that the distributions induced from $\tuple x$ and $\tuple y$
are identical.

The marginal distribution equivalence atom is defined in terms of
multisets of assignment weights. We distinguish multisets from sets by
using double wave brackets, e.g., $\{\{a,a,b\}\}$ denotes the multiset
$(\{a,b\},m)$ where $a$ and $b$ are given multiplicities $m(a)=2$ and
$m(b)=1$. If $\tuple x,\tuple y$ are variable sequences, then
$\tuple x \approx^* \tuple y$ is a \emph{marginal distribution
  equivalence atom} with the following semantics:
\begin{equation}\label{eq:marg}
  \A \models_{\X} \vec{x} \approx^* \vec{y} \Leftrightarrow \{\{\lvert {\X}_{\vec{x}=\vec{a}} \rvert > 0  \mid \tuple a \in A^{|\tuple x|}\}\}  
  = \{\{\lvert {\X}_{\vec{y}=\vec{b}} \rvert > 0 \mid \tuple b \in A^{|\tuple y|}\}\}.
\end{equation}
The next example illustrates the relationships between marginal distribution equivalence atoms and marginal identity atoms; the latter implies the former, but not vice versa.

\begin{figure}[t]
\centering
\begin{tabular}{@{\hspace{1.5em}}ccccc@{\hspace{1.5em}}}
\multicolumn{5}{c}{$\X$}\\\toprule
 $x$ & $y$ & $z$ &&  P  \\\midrule
 $a$ & $b$ & $c$  && $1/2$\\
 $b$ & $c$ & $b$ && $1/2$\\\bottomrule
\end{tabular}
\caption{A representation of a probabilistic team $\X$, for Example \ref{ex:atoms}, with domain $\{x,y,z\}$ that consists of two assignments whose probabilities are $1/2$.}\label{fig:ex1}
\label{fig:ex1}
\end{figure}

\begin{example}\label{ex:atoms}
Let $\X$ be the probabilistic team depicted in Figure \ref{fig:ex1}. The team $\X$ satisfies the atoms $xy\approxt y$, $x\approxt y$, $y\approxt z$, and $y\approx z$. The team $\X$ falsies the atom $x\approx y$, whereas $xy\approx y$ is not a well formed formula.
\end{example}

If $\tuple x,\tuple y,\tuple z$ are variable sequences, then $\pcixyz$
is a \emph{probabilistic conditional independence atom} with the
satisfaction relation defined as
\begin{align}
  \A\models_{\X} \pcixyz \label{def1}
\end{align}
if for all $s\colon\Var(\tuple x\tuple y\tuple z) \to A$ it holds that
\[
\lvert {\X}_{\tuple x\tuple y= s(\tuple x\tuple y)} \rvert \cdot
\lvert {\X}_{\tuple x\tuple z=s(\tuple x\tuple z)} \rvert = \lvert
{\X}_{\tuple x\tuple y\tuple z=s(\tuple x\tuple y\tuple z)} \rvert
\cdot \lvert {\X}_{\tuple x=s(\tuple x)} \rvert .
\]
Furthermore, we define \emph{probabilistic marginal independence atom}
$\pmixy$ as $\pci{\emptyset}{\tuple x}{\tuple y}$, i.e., probabilistic
independence conditioned by the empty tuple.

In addition to atoms based on counting or arithmetic operations, we
may also include all dependency atoms from the team semantics
literature. Let $\alpha$ be an atom that is interpreted in team
semantics, let $\A$ be a finite structure, and $\X:X\to [0,1]$ a
probabilistic team. We define $\A \models_\X \alpha$ if
$\A\models_{X^+} \alpha$, where $X^+$ consists of those assignments of
$X$ that are given positive weight by $\X$.  In this paper we will
discuss dependence atoms also in the context of probabilistic team
semantics. If $\tuple x,\tuple y$ are two variable sequences, then
$\dep(\tuple x,\tuple y)$ is a \emph{dependence atom} with team
semantics:
\begin{equation}\label{eq:marg}
  \A \models_{X} \dep(\tuple x,\tuple y) \Leftrightarrow s(\tuple x)= s'(\tuple x)\text{ implies } s(\tuple y)=s'(\tuple y)\text{ for all } s,s'\in X.
\end{equation}
A dependence atom of the form $\dep(\emptyset,\tuple x)$ is called a \emph{constancy atom}, written $\dep(\tuple x)$ in shorthand notation.
Dependence atoms can be expressed by using probabilistic independence
atoms. This has been shown for multiteams in
\cite{DurandHKMV18}, and the proof applies to probabilistic
teams.
\begin{proposition}[\cite{DurandHKMV18}]\label{prop:dep}
  Let $\A$ be a structure, $\X:X\to [0,1]$ a probabilistic team
  of $\A$, and $\tuple x$ and $\tuple y$ two sequences of
  variables. Then
  \(
  \A\models_\X \dep(\tuple x,\tuple y) \Leftrightarrow \A\models_\X
  \pci{\tuple x}{\tuple y}{\tuple y}.
  \)
\end{proposition}
Given a collection $C$ of atoms from
$\{\cpind,\pind,\approx,\approx^*,\dep(\cdot)\}$, we write $\FO(C)$
for the logic that extends $\FO$ with the atoms in $C$.

\begin{example}
Let $f_1, \ldots ,f_n,g$ be univariate distributions. Then $g$ is a \emph{finite mixture} of $f_1, \ldots ,f_n$ if it can be expressed as a convex combination of  $f_1, \ldots ,f_n$, i.e., if there are non-negative real numbers $r_1, \ldots ,r_n$ such that  $r_1+\ldots +r_n=1$ and $g(a)= \sum_{i=1}^n r_if_i(a)$. A probabilistic team $\X:X \to [0,1]$  gives rise to a univariate distribution $f_x(a) := |\X_{x=a}|$ for each variable $x$ from the domain of $X$. 
 The next formula expresses that the distribution $f_y$ is a finite mixture of the distributions $f_{x_1}, \ldots ,f_{x_n}$:
\[\exists qr \big[ \pmi{x_1\ldots x_n}{r}\wedge  \bigvee_{i=1}^n r=i \wedge  \bigwedge_{i=1}^n \exists x'r' \big(x_ir \approx x'r' \wedge [(q= i \vee r'= i) \to yq = x'r']\big)\big],
\]
where the indices $1, \ldots ,n$ are also thought of as distinct constants, and $(q= i \vee r'= i) \to yq = x'r'$ stands for $\neg  (q\neq i \wedge r'\neq i) \vee yq = x'r'$. The non-negative real numbers $r_i$ are represented by the weights of $r=i$ where $r$ is distributed independently of each $x_i$. The summand $r_if_{x_i}(a)$ is then represented by the weight of $x_ir=ai$ and $f_y(a)$ by the weight of $y=a$. The quantified subformula 
expresses that the former weight matches the weight of $yq=ai$, which implies that $f_y(a)$ is $r_1f_{x_1}(a) +\ldots +r_nf_{x_n}(a)$. 
\end{example}

\begin{example}
Probabilistic team semantics can be also used to model properties of data obtained from a quantum experiment  (adapting the approach of  \cite{Abrasmky_corr/abs-1007-2754}). Consider a probabilistic team $\X$ over variables $m_1,\dots,m_n, o_1,\dots ,o_n$. The intended interpretation of $\X(s)=r$ is that the joint probability that $s(m_i)$ was measured with outcome $s(o_i)$, for $1\leq i \leq m$, is $r$. In this setting many important properties of the experiment can be expressed using our formalism. For example the formula
\[
o_i \pind_{\vec{m}} (o_1,\dots,o_{i-1},o_{i+1},\dots, o_m)
\]
expresses a property called \emph{Outcome-Independence}; given the measurements $\vec{m}$, the outcome at $i$ is independent of the outcomes at other positions. The dependence atom $\dep(\vec{m}, \vec{o})$ on the other hand corresponds to a property called \emph{Weak-Determinism}. Moreover, if $\phi$ describes some property of hidden-variable models (Outcome-Independence, etc.), then the formula $\exists \lambda \phi$ expresses that the experiment can be explained by a hidden-variable model satisfying that property. 
\end{example}

The next example  relates probabilistic team semantics to  Bayesian networks. The example is an adaptation of an example discussed also in \cite{DurandHKMV18}.
\begin{example}\label{ex3}
Consider the Bayesian network $\mathbb{G}$ in Fig. \ref{network} that models beliefs about house safety using four Boolean random variables \texttt{thief}, \texttt{cat}, \texttt{guard} and \texttt{alarm}.  We refer to these variables by  $t,c,g,a$.  The dependence structure of a Bayesian network is characterized by the so-called local directed Markov property  stating  that each variable  is conditionally independent of its non-descendants given its parents. For our network $\mathbb{G}$ the only non-trivial independence given by this property is $ \pci{tc}{g}{a}$. Hence a joint distribution $P$ over  $t,c,g,a$ factorizes according to 
$\mathbb{G}$ if $\X$ satisfies $ \pci{tc}{g}{a}$. In this case $P$ can be factorized by
 \begin{equation}\label{exeq}
 P(t,c,g,a)=P(t)\cdot P(c\mid t)\cdot P(g\mid t,c)\cdot P(a\mid t,c)
 \end{equation}
  where, for instance, $t$ abbreviates either $\texttt{thief}=T$ or $\texttt{thief}=F$, and $P(c\mid t)$ is the probability of $c$ given $t$. The joint probability distribution (i.e., the team $\X$) can hence be stored as in Fig. \ref{network}.
Note that while  $\mathbb{G}$ expresses the independence statement $\pci{tc}{g}{a}$,  $\FOprob$-formulas  can be used to further refine the joint probability distribution as follows.
 Assume we have information suggesting that we may safely assume an $\FOprob$ formula $\phi$ on  $\X$:
\begin{itemize}
\item $\phi := t=F\to g=F$ indicates that \texttt{guard} never raises alarm in absence of \texttt{thief}. In this case the two bottom 
rows of the conditional probability distribution for \texttt{guard} become superfluous.
\item the assumption that $\phi$ is satisfied also exemplifies an interesting form of \emph{contex-specific} independence (CSI) that cannot be formalized by the usual  Bayesian networks (see, e.g.,  \cite{CoranderKH16b}). Namely, $\phi$ implies that  \texttt{guard} is independent of  \texttt{cat} in the context  $\texttt{thief}=F$. Interestingly such CSI statements can be formalized
utilizing the  disjunction of  $\FOprob$:
\[  t=T \vee   (t=F\wedge \pci{}{g}{c}).   \]
\item satisfaction of $\phi := tca\approx tcg$ would imply that \texttt{alarm} and \texttt{guard} have the same reliability for any given value of \texttt{thief} and \texttt{cat}. Consequently, the conditional distributions for \texttt{alarm} and \texttt{guard} are equal and one of the them could be removed.
\end{itemize}

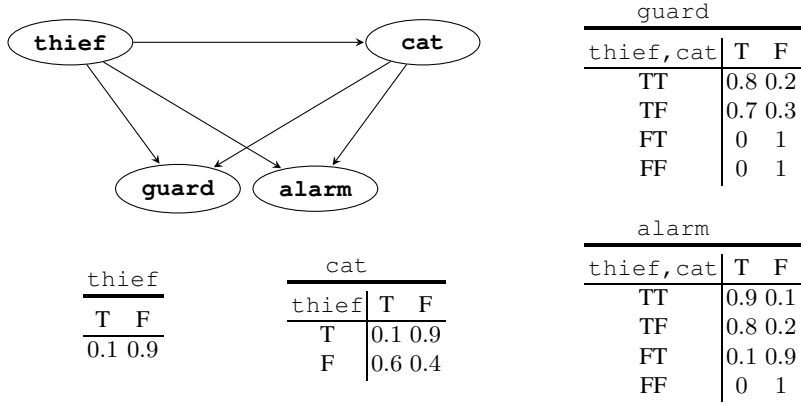
\begin{figure}[t]
\begin{tabular}{ccc}
\begin{tabular}{cc}
\multicolumn{2}{c}{
\begin{tikzpicture}[->,>=stealth, shorten >=.5pt,auto,node distance=2cm,
 main node/.style={circle,draw,ellipse,minimum size=0.6cm,minimum width=15mm,draw,font=\sffamily\bfseries}]
\usetikzlibrary{arrows}
\usetikzlibrary{shapes}
  \node[main node]  (1)  {\texttt{thief}};
 \node[main node, right=75pt]   (2) [right of=1] {\texttt{cat}};
   \node[main node, below=15pt]   (3) [below right  of=1] {\texttt{guard}};
   \node[main node, below=15pt]   (4) [below  left of=2] {\texttt{alarm}};
\path[every node/.style={sloped,anchor=south,auto=false}]
    (1) edge (2)
    (1) edge (3)
    (1) edge (4)
    (2) edge (4)
    (2) edge (3);
\end{tikzpicture}
}\\
{}\\
\mbox{\hspace{1cm}
\begin{tabular}{cc}
\multicolumn{2}{c}{\texttt{thief}}\\\toprule
T&F\\\hline
$0.1$&$0.9$
\end{tabular}}
&
\mbox{\hspace{1cm}
\begin{tabular}{c|cc}
\multicolumn{2}{c}{\texttt{cat}}\\\toprule
\texttt{thief}&T&F\\\hline
T&$0.1$&$0.9$\\
F&$0.6$&$0.4$
\end{tabular}}
\end{tabular}
&
\hspace{1cm}
&
\begin{tabular}{c}
\mbox{
\begin{tabular}{c|cc}
\multicolumn{2}{c}{\texttt{guard}}\\\toprule
\texttt{thief,cat}&T&F\\\hline
TT&$0.8$&$0.2$\\
TF&$0.7$&$0.3$\\
FT&$0$&$1$\\
FF&$0$&$1$\\
\end{tabular}}\\

{}\\

\mbox{
\begin{tabular}{c|cc}
\multicolumn{2}{c}{\texttt{alarm}}\\\toprule
\texttt{thief,cat}&T&F\\\hline
TT&$0.9$&$0.1$\\
TF&$0.8$&$0.2$\\
FT&$0.1$&$0.9$\\
FF&$0$&$1$\\
\end{tabular}}\\
\end{tabular}

\end{tabular}
\caption{Bayesian network $\mathbb{G}$ and its related conditional distributions\label{network}
}
\end{figure}
\end{example}

The following locality property dictates that satisfaction of a
formula $\phi$ in probabilistic team semantics depends only on the
free variables of $\phi$. For this, we define the \emph{restriction}
of a team $X$ to $V$ as
$X\upharpoonright V=\{s\upharpoonright V\mid s\in X\}$ where
$s\upharpoonright V$ denotes the restriction of the assignment $s$ to
$V$. The restriction of a probabilistic team $\X:X\to [0,1]$ to $V$ is then
defined as the probabilistic team
$\Y\colon X\upharpoonright V \to [0,1]$ where
$\Y(s)= \sum_{s'\upharpoonright V= s} \X(s').$ The set of \emph{free
  variables} $\Fr(\phi)$ of a formula over probabilistic team semantics 
   is defined
recursively as in first-order logic; note that for any atom $\phi$,
$\Fr(\phi)$ consists of all variables that appear in $\phi$.
\begin{proposition}[Locality, \cite{HKMV18}]\label{prop:locality}
  Let $\phi(\tuple x)\in \FO(\cpind,\approx,\approx^*,\dep(\cdot))$ be
  a formula with free variables from $\tuple x=(x_1, \ldots ,x_n)$.
  Then for all structures $\A$ and probabilistic teams $\X:X\to [0,1]$
  where $ \{x_1, \ldots ,x_n\} \subseteq V\subseteq\Dom(X)$,
  \(\A\models_{\X} \phi \iff \A\models_{\X\upharpoonright V}\phi.\)
\end{proposition}

Given two logics $\calL$ and $\calL'$ over probabilistic team semantics, we write $\calL\leq \calL'$ if for all open formulae $\phi(\tuple x)\in \calL$ there is a formula $\psi(\tuple x)\in \calL'$ such that $\A\models_\X \phi\Leftrightarrow \A\models_\X\psi$, for all structures $\A$ and probabilistic teams $\X$. The equality "$\equiv$" and strict inequality "$<$" relations between $\calL$ and $\calL'$ are defined from "$\leq$" in the standard way.

\paragraph{\textbf{Alternative Definition.}}
Probabilistic teams can also be defined as mappings
$\X:X \to \mathbb{R}_{\geq 0}$ that have no restriction for the total
sum of assignment weights, $\mathbb{R}_{\geq 0}$ being the set of all
non-negative reals. Probabilistic team semantics with respect to such
real weighted teams is then given exactly as in Definition
\ref{def:probsem}, except that we define disjunction without scaling:
\[\A\models_{\X} (\psi \lor \theta) \Leftrightarrow \A\models_\Y \psi
\text{ and } \A \models_\Z \theta \text{ for some $\Y,\Z$
  s.t. }\Y\sqcup \Z = \X,\]
where the union $\Y\sqcup \Z$ is defined such that
$(\Y\sqcup \Z)(s)= \Y(s) + \Z(s)$ for each $s$. Whether interpreting
probabilistic teams as probability distributions or just mappings from
assignments to non-negative reals does not make any difference in our
framework. Hence we write $\X: X \to [0,1]$ for a probabilistic
team that is a distribution such that $\sum_{s\in X} \X(s)=1$, and
$\X:X \to \mathbb{R}_{\geq 0}$ for a probabilistic team that is any
mapping from assignments to non-negative reals. A probabilistic team
of the former type is then a special case of that of the latter. We will use both notions and their associated semantics
interchangeably. If we need to distinguish between the two semantics,
we write $\models^{[0,1]}$ and $\models^{\geq 0}$ respectively for the
scaled (i.e., Definition \ref{def:probsem}) and non-scaled
variants. Given $\X:X \to \mathbb{R}_{\geq 0}$ and
$r\in \mathbb{R}_{\geq 0}$, we write $|\X|$ for the total weight
$\sum_{s\in X} \X(s)$ of $\X$, and $r\cdot \X$ for the probabilistic
team $\Y:X \to \mathbb{R}_{\geq 0}$ such that $\Y(s) = r\cdot \X(s)$
for all $s\in X$. The proposition below follows from a
straightforward induction (see Appendix \ref{a:prop1}).

\begin{repproposition}{prop}\label{lem:alt}
Let $\A$ be a structure, $\X:X\to \mathbb{R}_{\geq 0}$ a
  probabilistic team of $\A$,
  and
  $\phi \in \FO(\cpind,\approx,\approx^*,\dep(\cdot))$. Then
  \(
  \A\models^{\geq 0}_\X \phi \Leftrightarrow
  \A\models^{[0,1]}_{\frac{1}{|\X|} \cdot \X} \phi.
  \)
\end{repproposition}

\section{Expressiveness of $\FO(\pind)$}
Let $\X\colon X\to [0,1]$ be a probabilistic team
where $X$ is a finite set of assignements from a finite set
$D$ of variables.
%
A variable $x\in D$ is \emph{uniformly distributed} in $\X$ over a set of values $S$, if 
$$\X_{x=a}=\frac{1}{|S|}\text{ for all }a\in S\text{ and }X_{x=a}=0\text{ otherwise}.$$

The following lemma says essentially that if we can express constancy
and independence for a uniform distribution, then we can express~$\approx$.
Note that it may happen that we can express ``$\x$ uniformly distributed
and independent of~$\y$'' even when we cannot express ``$\x$ is independent
of~$\y$'' in general. For a proof of the lemma, see Appendix \ref{a:lemma}.


\begin{replemma}{lem}\label{lem:One}
Let $\A$ be structure with at least two elements and $\vec{z}$ an $n$-tuple of variables. Let
$\f(\z,d,c_1,c_2)$ be a formula
  such that for all probabilistic teams $\X$, whose variable domain includes  $\z,d,c_1,c_2$
  and for which  $\A\models_\X c_1\ne c_2$ and  $\A\models_\X \dep(c_1)\land \dep(c_2)$,
  it holds that
  \begin{align}\label{eq:0}
  \M\models_\X\f \quad\Leftrightarrow\quad &\text{$d$ is uniformly distributed over the two values of $c_1,c_2$}\\
  &\text{and $d$ is independent of~$\z$.} \notag
  \end{align}
  Then $\x\approx\y$ can be expressed for $n$-tuples $\x$ and $\y$
  using $\f$ and the constancy atom.
\end{replemma}

\begin{theorem}\label{thm:ind}
  $\FO(\approx)\leq \FO(\pind)$.
\end{theorem}
\begin{proof}

Proposition \ref{prop:dep} established that the constancy atom $\dep(x)$ can be equivalently expressed by the
independence atom $x \!\pind\! x$. Hence it is
  enough to show that we can define the formula $\f$ of Lemma \ref{lem:One} by using $\pind$.
  
  Let $\A$ and $\X$ be as assumed in Lemma \ref{lem:One}. We use below $\exists b\in \{c_1,c_2\}\, \theta$ as an abbreviation for $\exists b  (b=c_1 \lor b= c_2) \land \theta$, and $\forall  b\in \{c_1,c_2\}\, \theta$ for $\forall b   (b \neq c_1 \land b \neq c_2) \lor \big(  (b = c_1 \lor b = c_2) \land \theta \big)$.   Define $\f(\z,d,c_1,c_2)$ as
\[
  (\z\pind d)\land \forall a \in \{c_1, c_2\}
  \exists b \in \{c_1,c_2\}\big[(a\pind b)\land \big((a=b\land d=c_1)\lor (a\ne b\land
  d=c_2)\big)\big].
\]
  It suffices to prove (5). 
  The formula $\f$ clearly states that $\z$ and $d$ are independent. The formula also states that the values of $d$ range over the values of $c_1$ and $c_2$. It remains to be
  shown, conditioned on that $\z$ and $d$ are independent, that
  \[
  \text{$\A \models_\X \phi$ if and only if $d$ is uniformly distributed over $c_1$ and $c_2$.}
  \]
  Note that, by assumption of Lemma~\ref{lem:One},
  $c_1$ and $c_2$ are distinct constants.
  Let $\X_1$ be a team obtained from $\X$ by the quantification of $a$ and $b$.
  By the definition of universal quantification, in $\X_1$ $a$ is uniformly distributed and independent of everything else except maybe $b$. Note that 
  $d$  is uniformly distributed over the values of $c_1$ and $c_2$ in $X$ if and only if it is in $X_1$.

  If $d$
  is uniformly distributed over the values of $c_1$ and $c_2$, then picking values of $b$ with a uniform probability such that the right conjunct in
  \begin{equation}\label{eq:2}
  \big[(a\pind b)\land \big((a=b\land d=c_1)\lor (a\ne b\land
  d=c_2)\big)\big]
  \end{equation}
  holds clearly yields a team in which the left conjunct also holds.
  However, if $d$ is not uniformly distributed over  $c_1$ and $c_2$, then picking values for $b$ such that the right conjunct of \eqref{eq:2} holds 
  will yield $b$ that is not independent on $a$. \qed
\end{proof}

We also note that conditional independence is definable using marginal
independence. The proof applies ideas from \cite{HKMV18} and can be
found in Appendix \ref{a:thm}.
\begin{reptheorem}{thm}\label{thm:marg}
$\FO(\pind) \equiv \FO(\cpind) $.
\end{reptheorem}

%

\section{Expressiveness of $\FO(\approxt)$ and $\FO(\approx)$}
Initially it may seem that first-order logic with marginal distribution equivalence atoms is
less expressive than that with marginal identity atoms, as the former
atoms are given a strictly weaker truth condition. Contrary to this
intuition, however, we will in this section show that $\FO(\approxt)$
is actually strictly more expressive than $\FO(\approx)$. The result
is proven in two phases. First, in Sect. \ref{sect:approxt} we show
that dependence and marginal identity can be defined in
$\FO(\approxt)$, the former by a single marginal distribution equivalence atom and the latter by
a more complex formula.
Second, in Sect. \ref{sect:approx} we show
that the expressiveness of $\FO(\approx)$ is restricted by a union
closure property which is similar to that of inclusion logic in team
semantics. Since dependence atoms lack this property, the strict
inequality between $\FO(\approx)$ and $\FO(\approxt)$
follows. 

\subsection{Translations of Dependence and Marginal Identity to
  $\FO(\approxt)$}\label{sect:approxt}
We observe first that dependence atoms can be expressed in terms of
marginal distribution equivalence atoms, which in turn are definable using marginal identity and
dependence atoms. 
\begin{proposition}\label{prop:easytranslations}
  The following equivalences hold:
  \begin{enumerate}
  \item $\dep(\tuple x,y)\equiv \tuple xy\approxt \tuple x$,
  \item
    $\tuple x \approxt \tuple y \equiv \exists \tuple z (\dep(\tuple
    y,\tuple z) \wedge \dep(\tuple z,\tuple y)\wedge \tuple x \approx
    \tuple z)$.
  \end{enumerate}
\end{proposition}

Defining marginal identity atoms in $\FO(\approxt)$ is more
cumbersome.
Let $\X:X\to \mathbb{R}_{\geq 0}$ be a probabilistic team, and $\phi$
a quantifier-free first-order formula over the empty vocabulary (i.e.,
such that its satisfaction depends only on the variable
assignment). We define $\X_{\phi}:X\to \mathbb{R}_{\geq 0}$ as
the probabilistic team such that $\X_{\phi}(s) = \X(s)$ if $s$
satisfies $\phi$, and $\X_\phi(s)=0$ otherwise. Given two sequences of variables $\tuple x=(x_1,\ldots ,x_n)$ and $\tuple y= (y_1, \ldots ,y_n)$, we write $\tuple x \neq \tuple y$ as a shorthand for $\bigvee_{i=1}^n \neg x_i = y_i$.

\begin{theorem}\label{thm:incbyincs}
  $\tuple x\approx \tuple y$ is equivalent to $\phi\in \FO(\approxt)$
  where
  \[ \phi:=\forall \tuple z \big( (\tuple z\neq \tuple x \wedge \tuple
  z \neq \tuple y) \vee ((\tuple z= \tuple x \vee\tuple z=\tuple y)
  \wedge \tuple z \approxt\tuple x \wedge \tuple z \approxt \tuple
  y)\big).\]
\end{theorem}
\begin{proof}   Assume that
  $\tuple x,\tuple y, \tuple z$ are all $m$-ary. Let $\A$ be a
  structure with domain $A=\{1, \ldots ,n\}$, and let
  $\X: X \to \mathbb{R}_{\geq 0}$ a probabilistic team. Assume first
  that $\A \models_{\X} \tuple x\approx\tuple y$, that is, for all
  $\tuple i\in A^m$, the weights $|\X_{\tuple x=\tuple i}|$ and
  $|\X_{\tuple y=\tuple i}|$ coincide. It suffices to show that
  $\A \models_{\Y} \tuple z \approxt \tuple x \wedge\tuple z \approxt
  \tuple y$
  for $\Y:= \X'_{\theta}$ where $\theta$ is
  $ \tuple z=\tuple x\vee \tuple z=\tuple y$ and
  $\X'=\X[A^m/\tuple z]$ is the probabilistic team obtained from $\X$
  by distributing $A^m$ to $\tuple z$ uniformly.  For each
  $\tuple i\in A^m$ we consider three weight measures, obtained by
  dividing assignments associated with $\tuple i$ into three parts,
  $l_{\tuple i}:=|\X_{\tuple x=\tuple i\wedge \tuple x\neq \tuple
    y}|$,
  $r_{ \tuple i} :=|\X_{\tuple y= \tuple i\wedge \tuple x\neq \tuple
    y}|$,
  and
  $c_{\tuple i}:=|\X_{\tuple x=\tuple i \wedge \tuple y=\tuple i}|$.
  Then
  \[ |\Y_{\tuple x=\tuple i} |= |\X'_{\theta\wedge \tuple x= \tuple i
  }|= |\X'_{ \theta\wedge \tuple x=\tuple i\wedge \tuple x\neq \tuple
    y}| + |\X'_{\theta\wedge \tuple x= \tuple i\wedge \tuple y=\tuple
    i }|= \frac{2l_{\tuple i}+c_{\tuple i}}{n^m}.\]
  Observe that for
  $\X'_{ \theta\wedge \tuple x= \tuple i\wedge \tuple x\neq \tuple y}$
  we first partition each assignment in
  $\X_{\tuple x=\tuple i\wedge \tuple x\neq\tuple y}$ uniformly to
  $n^m$ parts in terms of the value of $\tuple z$ and then keep only
  those parts where $\theta$ holds.
  Since $\tuple x$ and $\tuple y$ disagree for every assignment in
  $\X'_{\tuple x=\tuple i\wedge \tuple x\neq \tuple y}$, the total
  weight of
  $\X'_{ \theta\wedge \tuple x=\tuple i\wedge \tuple x\neq \tuple y}$
  is obtained by multiplying $l_{\tuple i}$ with $\frac{2}{n^m}$. For
  $\X'_{ \theta\wedge \tuple x=\tuple i\wedge \tuple y=\tuple i}$ we
  have identical $\tuple x$ and $\tuple y$, and hence its weight is
  obtained by multiplying $c_{\tuple i}$ with $\frac{1}{n^m}$. By
  analogous reasoning we obtain that
  \[ |\Y_{\tuple y=\tuple i}|= \frac{2r_{\tuple i}+c_{\tuple i}}{n^m}
  \text{ and }|\Y_{\tuple z=\tuple i}|= \frac{r_{\tuple i}+l_{\tuple
      i}+c_{\tuple i}}{n^m}.\]
  Since our assumption implies $l_{\tuple i}=r_{\tuple i}$ for all
  $\tuple i$, the claim now follows from the observation that
  $\{\{|\Y_{\tuple u=\tuple i}|\mid \tuple i\in A^m\}\}$ are identical
  multisets for $\tuple u\in\{\tuple x,\tuple y,\tuple z\}$.
  
  Vice versa, assuming $\A \models_{\X} \phi$ we show
  $\A \models_{\X} \tuple x\approx \tuple y$. Let the weights
  $l_{\tuple i},r_{\tuple i},c_{\tuple i}$ and the probabilistic team
  $\Y$ be as above. By assumption we have
  $\A\models_\Y \tuple z \approxt \tuple x \wedge \tuple z \approxt
  \tuple y$, and
  thus 
  the following multisets are identical: 
  \begin{align*}
    &W_{\tuple x}:= \{\{2l_{\tuple 1}+c_{\tuple 1}, \ldots ,2l_{\tuple n}+c_{\tuple n}\}\},\\ 
    &W_{ \tuple y}:= \{\{2r_{\tuple 1}+c_{\tuple 1}, \ldots ,2r_{\tuple n}+c_{\tuple n}\}\}, \\
    &W_{\tuple z}:= \{\{l_{\tuple 1}+r_{\tuple 1} +c_{\tuple 1}, \ldots ,l_{\tuple n}+r_{\tuple n}+c_{\tuple n}\}\},
  \end{align*}
  where
    $\vec{1}=(1,\dots,1)$ and $\vec{n}=(n,\dots,n)$. 
  Assume to the contrary that
  $\A \not\models_{\X} \tuple x\approx \tuple y$, that is,
  $l_{\tuple i}\neq r_{\tuple i}$ for some $\tuple i$. Observe that
  whenever $l_{ \tuple j}=r_{ \tuple j}$ agree, then $\tuple j$
  contributes the same weight to all $W_{\tuple x}$, $W_{\tuple y}$,
  and $W_{\tuple z}$. Therefore, we may assume without loss of
  generality that $l_{\tuple i}\neq r_{\tuple i}$ for all $\tuple i$.
  Assume that $2l_{\tuple j}+c_{\tuple j}$ is the smallest element
  from $W_{\tuple x}$.  Since $W_{\tuple x}= W_{\tuple z}$, it follows
  that
  $2l_{\tuple j}+c_{\tuple j}= l_{\tuple k}+r_{\tuple k} +c_{\tuple
    k}$
  for some $\tuple k$. If $l_{\tuple k} < r_{\tuple k}$, then
  $2l_{\tuple k}+c_{\tuple k}< l_{\tuple k}+r_{\tuple k} +c_{\tuple
    k}$
  which contradicts the assumption that $2l_{\tuple j}+c_{\tuple j}$
  is smallest. Since $W_{\tuple x} = W_{\tuple y}$, similar
  contradiction follows from $r_{\tuple k}< l_{\tuple k}$, too. Hence,
  $\A \models_{\X} \tuple x\approx \tuple y$ which concludes the
  proof. \qed
\end{proof}

The following theorem now combines the results of this
section. Note that the translations to both directions are of linear size. 

\begin{theorem}\label{theorem:tahti}
  $\FO(\approxt)\equiv \FO(\approx,\dep(\cdot))$.
\end{theorem}

\subsection{Scaled Union Closure of $\FO(\approx)$}\label{sect:approx}
Inclusion logic is known to be union closed over teams. This means
that for all structures $\A$, teams $X$, and inclusion logic formulae
$\phi$: if $\A\models_X \phi$ and $\A \models_Y \phi$, then
$\A \models_{X\cup Y} \phi$.  The following proposition, proven in Appendix \ref{a:prop2}, demonstrates that $\FO(\approx)$ is endowed with an analogous closure
property, namely, that all formulae of $\FO(\approx)$ are closed under
all $k$-scaled unions of probabilistic teams.
\begin{repproposition}{prop:ucl}\label{prop:uc}
  Let $\A$ be a model, $\phi\in \FO(\approx)$ a formula, and
  $\X:X\to [0,1]$ and $\Y:X\to [0,1]$ two probabilistic teams. Then
  for all $k\in [0,1]$:
  \[\text{if }\A\models_\X \phi\text{ and }\A\models_\Y \phi\text{,
    then } \A\models_{\X \kcup \Y} \phi.\]
\end{repproposition}
\vspace{-.7cm}

As a corollary we observe that $\FO(\approx)$ is strictly weaker than
$\FO(\approx^*)$. Recall from Proposition \ref{prop:easytranslations}
that the constancy atom $\dep(x)$ is definable in
$\FO(\approx^*)$. However, constancy is clearly not preserved under
$k$-scaled unions, therefore falling outside the scope of
$\FO(\approx)$. Furthremore, by Theorem
\ref{thm:incbyincs} $\FO(\approx^*)$ is at least as expressive as
$\FO(\approx)$.

\begin{corollary}
  $\FO(\approx)< \FO(\approx^*)$.
\end{corollary}

\section{Binary Probabilistic Teams}

In this section we restrict attention to binary probabilistic teams
and propositional logic extended with quantifiers (see \cite{HLKV16} for related work). We define the
syntax of \emph{quantified propositional logic} $\QPL$ by the
following grammar
\begin{equation}\label{prop:synt}
  \phi::= p\mid \neg p \mid \phi\vee \phi \mid \phi \wedge \phi  \mid \exists p \phi \mid \forall p \phi,
\end{equation}
where $p$ is a proposition variable. The probabilistic team semantics
of $\QPL$ is defined analogously to that of first-order formulae. We
say that a probabilistic team $\X:X\to [0,1]$ is \emph{binary} if $X$
assigns variables into $\{0,1\}$. For a $\QPL$ formula $\phi$ and a
binary probabilistic team $\X:X \to [0,1]$, we write $\X\models \phi$
iff $\A \models_\X \phi^*$, where $\phi^*$ is the first-order formula
obtained from $\phi$ by substituting $P(p)$ for $p$ and $\neg P(p)$ for $\neg p$,
and letting $\A := (\{0,1\},P^\A:=\{1\})$. Furthermore, we denote classical
negation by "$\sim$". That is, we write $\X \models \negg \phi$ if
$\X \not\models \phi$. We let $\QPL(\sim)$ denote the logic obtained
by the grammar \eqref{prop:synt} extended with $\negg \phi$, and
denote by $\QPL(\sim,C)$ the extension of $\QPL(\sim)$ by any collection of
dependencies $C$. 

We observe that $\QPL(\sim,\cpind,\approx)$ can be interpreted as
statements of real arithmetic. As truth in real arithmetic is
decidable, this gives us some fairly conservative upper bounds with
respect to the complexity of satisfiability and validity of
$\QPL(\sim,\cpind,\approx)$. We say that
$\phi\in \QPL(\sim,\cpind,\approx)$ is \emph{satisfiable} if $\phi$ is
satisfied by some non-empty binary probabilistic team.\footnote{Empty team satisfies every formula without $\sim$; with $\sim$ it is a non-interesting special case \cite{HannulaKVV18}.} Also, $\phi$
is \emph{valid} is $\phi$ is satisfied by all binary probabilistic
teams. Note that the \emph{free variables} of a $\QPL(\sim,C)$ formula are defined analogously to the first-order case.

\begin{theorem}\label{thm:arit}
  For each 
  $\phi \in \QPL(\sim,\cpind)$ ($\phi \in \QPL(\sim,\approx)$, resp.)
  there exists a first-order sentence $\psi$ over vocabulary
  $\{+,\times,\leq,0,1\}$ ($\{+,\leq,0\}$, resp.) such that $\phi$
  is satisfiable iff $(\mathbb{R},+,\times,\leq ,0,1)\models \psi$
  ($(\mathbb{R},+,\leq,0)\models \psi$, resp.).
\end{theorem}
\begin{proof}
  We show that satisfiability of a formula $\phi \in \QPL(\sim,\cpind)$ is definable in real arithmetic in terms of the non-scaled variant of probabililistic team semantics. 
  For a given tuple $\tuple p = (p_1, \ldots ,p_n)$ of proposition variables, we introduce fresh first-order variables
  $s_{\tuple p=\tuple i}$ for each propositional assignment $s(\tuple p)=\tuple i$,
  where $\tuple i$ is a binary string of length $n$. We write $\tuple s$ to denote the complete tuple of these variables.
 For a $\tuple p$ listing the free variables of $\phi$, we define
  \[\psi:= \exists s_{\tuple p =\tuple 0} \ldots s_{\tuple p =\tuple
    1} \big( \bigwedge_{\tuple i} 0 \leq s_{\tuple p=\tuple i}\ \wedge \neg 0=\sum_{\vec{i}} s_{\tuple p=\tuple i}   \wedge
  \phi^*(\tuple s) \big)
  \]
  where the mapping $\phi(\tuple p)\mapsto \phi^*(\tuple s)$ is
  defined recursively as follows:
  \begin{itemize}
  \item If $\phi(\tuple p)$ is a propositional literal, then
    $\phi^*(\tuple s):= \bigwedge_{s\not\models \phi} s=0$.
  \item If $\phi(\tuple p)$ is
    $\pci{\tuple a}{\tuple b}{\tuple c}$, where
    $\tuple p = \tuple a\tuple b\tuple c \tuple d$ for some
    $\tuple d$, then  $\phi^*(\tuple s)$ is defined as
    \begin{align*}
     \bigwedge_{\tuple i\tuple j \tuple k}  (
      & \sum_{\tuple l'} s_{\tuple a\tuple b\tuple c\tuple d=\tuple i\tuple j\tuple k\tuple l'}
        \times \sum_{\tuple j'\tuple k'\tuple l'} s_{\tuple a\tuple b\tuple c\tuple d=\tuple i\tuple j'\tuple k'\tuple l'}=  
      \sum_{\tuple k'\tuple l'} s_{\tuple a\tuple b\tuple c\tuple d=\tuple i\tuple j\tuple k'\tuple l'}\times
        \sum_{\tuple j'\tuple l'} s_{\tuple a\tuple b\tuple c\tuple d=\tuple i\tuple j'\tuple k\tuple l'}).
    \end{align*}
    
  \item If $\phi(\tuple p)$ is $\tuple a\approx \tuple b$, where     $\tuple p = \tuple a\tuple b\tuple c $ for some $\tuple c$, then
    \begin{align*}
      \phi^*(\tuple s):= \bigwedge_{\tuple i} \sum_{\tuple j'\tuple k'} s_{\tuple a\tuple b\tuple c=\tuple i\tuple j'\tuple k'}= 
      \sum_{\tuple j'\tuple k'} s_{\tuple a\tuple b\tuple c=\tuple j'\tuple i\tuple k'}.
    \end{align*}
    
  \item If $\phi(\tuple p)$ is $\negg \eta(\tuple p)$, then
    $\phi^*(\tuple s):= \neg \eta^*(\tuple s)$.
  \item If $\phi(\tuple p)$ is
    $\eta(\tuple p) \wedge \chi(\tuple p)$, then
    $\phi^*(\tuple s):= \eta^*(\tuple s) \wedge \chi^*(\tuple s)$.
    
  \item If $\phi(\tuple p)$ is
    $\eta(\tuple p) \vee \chi(\tuple p)$, then
    \begin{align*}
      \phi^*(\tuple s):= \exists t_{\tuple p=\tuple 0}r_{\tuple p=\tuple 0}\ldots t_{\tuple p=\tuple 1}r_{\tuple p=\tuple 1}\big(
      &\bigwedge_{\tuple i}(0\leq t_{\tuple p=\tuple i} \wedge 0\leq r_{\tuple p=\tuple i} \wedge\\
      & s_{\tuple p=\tuple i}= t_{\tuple p=\tuple i}+ r_{\tuple p=\tuple i})\wedge 
        \eta^*(\tuple t) \wedge \chi^*(\tuple r)\big).
    \end{align*}
  \item If $\phi(\tuple p)$ is $\exists q \eta(\tuple p,q)$, then
    \begin{align*}
      \phi^*(\tuple s):= \exists t_{\tuple pq=\tuple 00}\ldots t_{\tuple pq=\tuple 11}\big(\bigwedge_{\tuple ij}(0\leq t_{\tuple pq=\tuple ij}   
      \wedge s_{\tuple p=\tuple i}= t_{\tuple p=\tuple i0}+ t_{\tuple p=\tuple i1})\wedge 
      \eta(\tuple t) \big).
    \end{align*}
  \item If $\phi(\tuple p)$ is $\forall y \eta(\tuple p,q)$, then
    \begin{align*}
      \phi^*(\tuple s):= \exists t_{\tuple pq=\tuple 00}\ldots t_{\tuple pq=\tuple 11}\big(&\bigwedge_{\tuple ij}(0\leq t_{\tuple pq=\tuple ij}   
                                                                                             \wedge s_{\tuple p=\tuple i}= t_{\tuple p=\tuple i0}+ t_{\tuple p=\tuple i1}\wedge \\
                                                                                           &t_{\tuple p=\tuple i0}= t_{\tuple p=\tuple i1})\wedge 
                                                                                             \eta(\tuple t) \big).
    \end{align*}
    It is straightforward to check that the claim follows. \qed
  \end{itemize}
\end{proof}

From the translation above we immediately obtain some complexity
bounds for the satisfiability and validity problems of quantified
propositional logics over probabilistic team semantics. We write $\twoEXPSPACE$ for the class of problems solvable in space $O(2^{2^{p(n)}})$, and $\AEXPTIME(f(n))$ ($\AtwoEXPTIME(f(n))$, resp.) for the class of problems solvable by alternating Turing machine in time $O(2^{p(n)})$ ($O(2^{2^{p(n)}})$, resp.) with $f(n)$ many alternations, where $p$ is a polynomial.
\begin{theorem}\label{thm:comp}
  The satisfiability/validity problems of the logics $\QPL(\cpind,\sim)$ and
  $\QPL(\approx,\sim)$ are in $\twoEXPSPACE$ and
  $\AtwoEXPTIME(2^{O(n)})$, respectively. 
\end{theorem}
\begin{proof}
  By the proof of Theorem \ref{thm:arit}, satisfiability and validity
  of quantified propositional formulae can be reduced to truth of a
  real arithmetic sentence of size $2^{O(n)}$. The stated upper bounds
  for $\QPL(\sim,\cpind)$ and $\QPL(\sim,\approx)$ then follow because the
  theory of real-closed fields, $\Th(\mathbb{R},+,\times,\leq,0,1)$,
  is in $\EXPSPACE$ \cite{BENOR1986251}, and the theory of real
  addition, $\Th(\mathbb{R},+,\leq,0)$, is in $\AEXPTIME(n)$
  \cite{BERMAN198071,FerranteR75}. \qed
\end{proof}

We also obtain an upper bound for the implication problem of conditional
independence 
  over binary probability distributions. The \emph{implication problem} for conditional independence is given as a finite set $\Sigma\cup\{\sigma\}$ of conditional independence statements, and the problem is to decide whether all probability distributions that satisfy $\Sigma$ satisfy also $\sigma$. It is a famous
open problem to determine whether implication of conditional
independence is decidable over discrete 
distributions. Since binary probabilistic teams can be interpreted as discrete 
distributions of binary random variables, we obtain that the
implication problem for conditional independence statements is
decidable in exponential space over binary distributions. The result follows since any instance of such an implication problem can be expressed as an existential formula of exponential size (Theorem \ref{thm:arit}), and since the existential theory of real-closed fields is in $\PSPACE$ \cite{Canny:1988}.
\begin{corollary}
  The implication problem for conditional independence 
   over binary probability distributions is in
  $\EXPSPACE$. 
\end{corollary}

It may be conjectured that the obtained complexity bounds are not
optimal. The first-order translations provide only access to a very
restricted type of arithmetic expressions. For instance, real
multiplication is only available between sums of reals from the unit
interval. We leave it as an open problem to determine whether the
results of this section can be optimized using more refined arguments.

\section{Conclusions and further directions}

\begin{table}[t]
\centering
\begin{tabular}{ll}
\toprule
PTS:&$\FO(\approx)< \FO(\approx,\dep(\cdot))\equiv  \FO(\approx^*) \leq \FO(\pind)\equiv \FO(\cpind)$\\
TS:&$\FO(\sub)< \FO(\sub, \dep(\cdot))\equiv \FO(\bot)\equiv \FO(\bot_{\rm c})$ \cite{galliani12,GallianiV14}\\
\bottomrule\\
\end{tabular}
\caption{Relative expressivity in probabilistic team semantics (PTS) and team semantics (TS)}\label{tbl:eka}
\end{table}

We have studied probabilistic team semantics in association with three notions of dependency atoms: probabilistic independence, marginal identity, and marginal distribution equivalence atoms. Our investigations give rise to an overall classification that is already familiar from the team semantics context (see Table \ref{tbl:eka}). Similar to inclusion logic ($\FO(\sub)$) in team semantics, we observed that $\FO(\approx)$ enjoys a union closure property which renders it strictly less expressive than $\FO(\approx,\dep(\cdot))$. A further analogous fact is that both dependence and marginal identity are definable with conditional independence, which in turn is definable using only marginal independence. An interesting open question is to determine the relationship between $\FO(\approx,\dep(\cdot))$  (or equivalently  $\FO(\approx^*)$) and $\FO(\cpind)$. Contrary to the picture arising from team semantics, we conjecture that the latter is strictly more expressive.

 One motivation behind our marginal distribution equivalence atom was that
it seemed to be weaker than marginal identity but still enough to
guarantee the same entropy of two distributions. A natural next step would
be to consider some form of entropy atom/atoms and study the expressive
power of the resulting logics. The exact formulation of such atoms will
make all the difference, as one can detect both functional dependencies
and marginal independence if one has full access to the conditional
entropy as a function.

We also studied (quantified) propositional logics with probabilistic team semantics. By connecting real-valued probabilistic teams to real arithmetic we showed upper bounds for computational problems associated with these logics. 
 As a consequence of our translation to real arithmetic we also obtained an $\EXPSPACE$ upper bound for the implication problem of conditional independence statements over binary distributions.

%
%
%
\bibliographystyle{splncs04} \bibliography{biblio}

\begin{thebibliography}{10}
\providecommand{\url}[1]{\texttt{#1}}
\providecommand{\urlprefix}{URL }
\providecommand{\doi}[1]{https://doi.org/#1}

\bibitem{Abrasmky_corr/abs-1007-2754}
Abramsky, S.: Relational hidden variables and non-locality. Studia Logica
  \textbf{101}(2),  411--452 (2013)

\bibitem{EPTCS286.2}
Barbero, F., Sandu, G.: Interventionist counterfactuals on causal teams. In:
  Finkbeiner, B., Kleinberg, S. (eds.) {\rm Proceedings 3rd Workshop on} formal
  reasoning about Causation, Responsibility, and Explanations in Science and
  Technology, {\rm Thessaloniki, Greece, 21st April 2018}. Electronic
  Proceedings in Theoretical Computer Science, vol.~286, pp. 16--30. Open
  Publishing Association (2019). \doi{10.4204/EPTCS.286.2}

\bibitem{BENOR1986251}
Ben-Or, M., Kozen, D., Reif, J.: The complexity of elementary algebra and
  geometry. Journal of Computer and System Sciences  \textbf{32}(2),  251 --
  264 (1986)

\bibitem{BERMAN198071}
Berman, L.: The complexity of logical theories. Theoretical Computer Science
  \textbf{11}(1),  71 -- 77 (1980)

\bibitem{Canny:1988}
Canny, J.: Some algebraic and geometric computations in pspace. In: Proceedings
  of the Twentieth Annual ACM Symposium on Theory of Computing. pp. 460--467.
  STOC '88, ACM, New York, NY, USA (1988)

\bibitem{Cavallo:1987}
Cavallo, R., Pittarelli, M.: The theory of probabilistic databases. In:
  Proceedings of the 13th International Conference on Very Large Data Bases.
  pp. 71--81. VLDB '87, Morgan Kaufmann Publishers Inc., San Francisco, CA, USA
  (1987)

\bibitem{CoranderKH16b}
Corander, J., Hyttinen, A., Kontinen, J., Pensar, J.,
  V{\"{a}}{\"{a}}n{\"{a}}nen, J.: A logical approach to context-specific
  independence. In: V{\"{a}}{\"{a}}n{\"{a}}nen, J.A., Hirvonen, {\AA}.,
  de~Queiroz, R.J.G.B. (eds.) Logic, Language, Information, and Computation -
  23rd International Workshop, WoLLIC 2016, Puebla, Mexico, August 16-19th,
  2016. Proceedings. Lecture Notes in Computer Science, vol.~9803, pp.
  165--182. Springer (2016). \doi{10.1007/978-3-662-52921-8\_11}

\bibitem{DurandHKMV18}
Durand, A., Hannula, M., Kontinen, J., Meier, A., Virtema, J.: Approximation
  and dependence via multiteam semantics. Ann. Math. Artif. Intell.
  \textbf{83}(3-4),  297--320 (2018),
  \url{https://doi.org/10.1007/s10472-017-9568-4}

\bibitem{HKMV18}
Durand, A., Hannula, M., Kontinen, J., Meier, A., Virtema, J.: Probabilistic
  team semantics. In: FoIKS. Lecture Notes in Computer Science, vol. 10833, pp.
  186--206. Springer (2018). \doi{10.1007/978-3-319-90050-6\_11}

\bibitem{FerranteR75}
Ferrante, J., Rackoff, C.: A decision procedure for the first order theory of
  real addition with order. {SIAM} J. Comput.  \textbf{4}(1),  69--76 (1975).
  \doi{10.1137/0204006}

\bibitem{galliani08}
Galliani, P.: {G}ame {V}alues and {E}quilibria for {U}ndetermined {S}entences
  of {D}ependence {L}ogic (2008), {MSc} Thesis. ILLC Publications,
  {MoL}--2008--08

\bibitem{galliani12}
Galliani, P.: Inclusion and exclusion dependencies in team semantics: On some
  logics of imperfect information. Annals of Pure and Applied Logic
  \textbf{163}(1),  68 -- 84 (2012)

\bibitem{GallianiV14}
Galliani, P., V{\"{a}}{\"{a}}n{\"{a}}nen, J.: On dependence logic. In: Baltag,
  A., Smets, S. (eds.) Johan van Benthem on Logic and Information Dynamics, pp.
  101--119. Springer (2014). \doi{10.1007/978-3-319-06025-5\_4}

\bibitem{GradelG98}
Gr{\"{a}}del, E., Gurevich, Y.: Metafinite model theory. Inf. Comput.
  \textbf{140}(1),  26--81 (1998). \doi{10.1006/inco.1997.2675}

\bibitem{gradel10}
Gr\"adel, E., V\"a\"an\"anen, J.: Dependence and independence. Studia Logica
  \textbf{101}(2),  399--410 (2013). \doi{10.1007/s11225-013-9479-2}

\bibitem{HannulaK16}
Hannula, M., Kontinen, J.: A finite axiomatization of conditional independence
  and inclusion dependencies. Inf. Comput.  \textbf{249},  121--137 (2016).
  \doi{10.1016/j.ic.2016.04.001}

\bibitem{HLKV16}
Hannula, M., Kontinen, J., L{\"{u}}ck, M., Virtema, J.: On quantified
  propositional logics and the exponential time hierarchy. In: GandALF.
  {EPTCS}, vol.~226, pp. 198--212 (2016)

\bibitem{HannulaKV18}
Hannula, M., Kontinen, J., Virtema, J.: Polyteam semantics. In: Logical
  Foundations of Computer Science - International Symposium, {LFCS} 2018,
  Deerfield Beach, FL, USA, January 8-11, 2018, Proceedings. pp. 190--210
  (2018). \doi{10.1007/978-3-319-72056-2\_12}

\bibitem{HannulaKVV18}
Hannula, M., Kontinen, J., Virtema, J., Vollmer, H.: Complexity of
  propositional logics in team semantic. {ACM} Trans. Comput. Log.
  \textbf{19}(1),  2:1--2:14 (2018). \doi{10.1145/3157054}

\bibitem{hodges97}
Hodges, W.: {C}ompositional {S}emantics for a {L}anguage of {I}mperfect
  {I}nformation. Journal of the Interest Group in Pure and Applied Logics
  \textbf{5 (4)},  539--563 (1997)

\bibitem{Hyttinen15b}
Hyttinen, T., Paolini, G., V{\"a}{\"a}n{\"a}nen, J.: {A Logic for Arguing About
  Probabilities in Measure Teams}. Arch. Math. Logic  \textbf{56}(5-6),
  475--489 (2017). \doi{10.1007/s00153-017-0535-x}

\bibitem{kmvz18}
Krebs, A., Meier, A., Virtema, J., Zimmermann, M.: {Team Semantics for the
  Specification and Verification of Hyperproperties}. In: Potapov, I.,
  Spirakis, P., Worrell, J. (eds.) 43rd International Symposium on Mathematical
  Foundations of Computer Science (MFCS 2018). Leibniz International
  Proceedings in Informatics (LIPIcs), vol.~117, pp. 10:1--10:16. Schloss
  Dagstuhl--Leibniz-Zentrum fuer Informatik, Dagstuhl, Germany (2018).
  \doi{10.4230/LIPIcs.MFCS.2018.10}

\bibitem{Luck18}
L{\"{u}}ck, M.: Canonical models and the complexity of modal team logic. In:
  27th {EACSL} Annual Conference on Computer Science Logic, {CSL} 2018,
  September 4-7, 2018, Birmingham, {UK}. pp. 30:1--30:23 (2018).
  \doi{10.4230/LIPIcs.CSL.2018.30}

\bibitem{NiepertGSG13}
Niepert, M., Gyssens, M., Sayrafi, B., Gucht, D.V.: On the conditional
  independence implication problem: {A} lattice-theoretic approach. Artif.
  Intell.  \textbf{202},  29--51 (2013). \doi{10.1016/j.artint.2013.06.005}

\bibitem{vaananen07}
V\"a\"an\"anen, J.: Dependence Logic. Cambridge University Press (2007)

\end{thebibliography}

\appendix

\section{Proof of Proposition \ref{lem:alt}}\label{a:prop1}

\repeatproposition{prop}
\begin{proof}  The cases for first-order literals, $\approx$,
  $\approx^*$, $\dep(\cdot)$ and the conjunction are immediate. The claim
  for the independence atom $\pcixyz$ follows from the equivalence
  below together with the observation that the former is the
  definition of the atom in the unscaled team $\X$ whereas the latter
  is equivalent to that of the scaled team $\frac{1}{|\X|} \cdot \X$.
  \begin{multline*}
    \lvert {\X}_{\tuple x\tuple y= s(\tuple x\tuple y)} \rvert \cdot
    \lvert {\X}_{\tuple x\tuple z=s(\tuple x\tuple z)} \rvert = \lvert
    {\X}_{\tuple x\tuple y\tuple z=s(\tuple x\tuple y\tuple z)} \rvert
    \cdot \lvert {\X}_{\tuple x=s(\tuple x)} \rvert
    \text{ if and only if }
    \\
    \frac{1}{|\X|} \cdot \lvert {\X}_{\tuple x\tuple y= s(\tuple
      x\tuple y)} \rvert \cdot \frac{1}{|\X|} \cdot \lvert
    {\X}_{\tuple x\tuple z=s(\tuple x\tuple z)} \rvert =
    \frac{1}{|\X|} \cdot \lvert {\X}_{\tuple x\tuple y\tuple
      z=s(\tuple x\tuple y\tuple z)} \rvert \cdot \frac{1}{|\X|} \cdot
    \lvert {\X}_{\tuple x=s(\tuple x)} \rvert.
  \end{multline*}
  The case for disjuction follows from the following chain of
  equivalences
  \begin{align*}
    \A \modelsr_\X \phi \lor \psi &\Leftrightarrow \A \modelsr_\Y \phi \text{ and } \A \modelsr_\Z \psi \text{ for some $\Y$ and $\Z$ s.t. $\Y\sqcup \Z = \X$}\\
                                  &\Leftrightarrow \A \modelsd_{\frac{1}{\lvert \Y \rvert} \cdot \Y} \phi 
                                    \text{ and } \A \modelsd_{\frac{1}{\lvert \Z \rvert} \cdot \Z} \psi \text{ for some $\Y$ and $\Z$ s.t. $\Y\sqcup \Z = \X$}\\
                                  &\Leftrightarrow \A \modelsd_{\frac{1}{\lvert \X \rvert} \cdot \X} \phi \lor \psi,
  \end{align*}
  where the last equivalence follows form the definition of the
  disjunction for $k=\frac{\lvert \Y \rvert}{\lvert \X\rvert}$ and
  $1-k = \frac{\lvert \Z \rvert}{\lvert \X\rvert}$, since
  \begin{align*}
    \frac{\lvert \Y \rvert}{\lvert \X\rvert} \cdot \frac{1}{\lvert \Y\rvert} \cdot \Y 
    + \frac{\lvert \Z \rvert}{\lvert \X\rvert} \cdot \frac{1}{\lvert \Z\rvert} \cdot \Z 
    = \frac{1}{\lvert \X\rvert} \cdot \Y + \frac{1}{\lvert \X\rvert} \cdot \Z =  \frac{1}{\lvert \X\rvert} \cdot \X.
  \end{align*}
  The cases for the quantifiers are similar; we show the case for the
  universal quantifier
  \begin{align*}
    \A \modelsr_\X \forall x \phi \Leftrightarrow \A \modelsr_{\X[A/x]} \phi 
    \Leftrightarrow \A \modelsd_{\frac{1}{\vert \X[A/x] \rvert} \cdot \X[A/x]} \phi 
    & \Leftrightarrow \A \modelsd_{(\frac{1}{\vert \X \rvert} \cdot\X) [A/x]} \phi \\
    & \Leftrightarrow  \A \modelsd_{\frac{1}{\vert \X \rvert} \cdot \X} \forall x \phi,
  \end{align*}
  where the second last equivalence follows, since
  $\vert \X[A/x] \rvert = \lvert \X \rvert$ and
  $(\frac{1}{\vert \X \rvert} \cdot\X) [A/x]= \frac{1}{\vert \X
    \rvert} \cdot\X [A/x]$.\qed
\end{proof}

\section{Proof of Lemma \ref{lem:One}}\label{a:lemma}

\repeatlemma{lem}

\begin{proof}
  We will write a formula $\psi(\x,\y)$ which is to be equivalent with
  $\x\approx \y$.  But first we need to define an auxiliary formula~$\theta$.
Define
  $$\theta := (d=c_1\land \z=\x) \lor (d=c_2\land \z=\y).$$
  This formula says that $\z$ always equals either $\x$ or $\y$ and
  $d$ is a ``detector'' for which one it is.
  We use the abbreviation $\exists^c c_1 c_2$ below to denote
  $\exists c_1 \exists c_2 (\dep(c_1)\land \dep(c_2) \land c_1\neq c_2)$.
  %
  Now define
  $$\psi(\x,\y) := \exists^c c_1 c_2 
  \Bigg[
      \forall \z\exists d
      \Big(
          (\x=\y)
          \lor 
          \big[
              (\x\ne \y) \land 
              \big(
                 (\z\ne \x\land \z\ne \y) \lor [\theta \land \f]
              \big)
          \big] 
      \Big) 
   \Bigg].$$

   Suppose $\x\approx \y$ holds in a team $\X$ over variables $\x$ and
   $\y$. We want to show that $\psi(\x,\y)$ is satisfied by~$\X$. 
   Let $\X_1$ be the expansion of $\X$ obtained by the quantification of $c_1$, $c_2$, and  $\z$.
  We may assume that $c_1$, $c_2$ were picked such that they
   attain constant but distinct values. Also note that $\z$
   is independent of all other variables and uniformly distributed over the domain of~$\A$.
   Now let $d$ be a variable that takes its values from the values of $c_1$ and $c_2$
  such that it ``detects''
   whether $\z$ equals $\x$ or not (value of $d$ is the value of
   $c_1$ iff $\z$ and $\x$ have the same value). Let $\X_2$ be the
   expansion of $\X_1$ by this $d$. We need to check that $\X_2$
   satisfies
   $$(\x=\y)
     \lor 
     \big[
         (\x\ne \y)\land 
         \big(
            (\z\ne \x\land \z\ne \y)\lor [\theta \land \f]
         \big)
     \big].$$
  Let $\X_3$ be the maximal subteam of $\X_2$ where $\x\ne \y$. So now we
  have to check that
  \begin{equation}\label{eq:1}
  (\z\ne \x\land \z\ne \y)\lor [\theta \land \f]
  \end{equation}
  holds in $\X_3$. Recall that $\theta$ says in particular that $\z$
  equals either $\x$ or $\y$, so \eqref{eq:1} holds in $\X_3$ if and only if
  $\theta\land\f$ holds in the maximal subteam $\X_4$ of $\X_3$ in
  which this is the case.  We also just defined $d$ to attain the
  value $c_1$ if and only if $\z=\x$ and the only other option is that
  $\z=\y$ in which case $d=c_2$, so
  $\theta$ is satisfied.  What about
  $\f$; note that $\X_4$ is such that \eqref{eq:0} holds. Now fix any value $\tuple v$ of $\z$ in
  $\X_4$. Since $\x\approx \y$ holds, we have
  $|\X_{\x=\tuple v}|=|\X_{\y=\tuple v}|$. When we expand $\X$ to
  $\X_1$ and further to $\X_2$ this property is (clearly)
  preserved. It is also preserved when we take the subteam $\X_3$,
  because when we move from $\X_2$ to $\X_3$, we only remove
  assignments $s$ where $s(\x)=s(\y)$, so if an assignment with
  $\x=\tuple v$ is deleted, then also an assignment with
  $\tuple y=\tuple v$ is deleted (the same assignment). When we move
  to $\X_4$ we still have $|(\X_4)_{\x=\tuple v}|=|(\X_4)_{\y=\tuple v}|$ which
  follows from the fact that $\z$ is independent of
  $\x,\y,c_1,c_2$. Therefore
  $$|(\X_4)_{\x\z=\tuple v\tuple v)}|=|(\X_4)_{\y\z)=\tuple v\tuple v}|.$$
  But this means that conditioned on $\z=\tuple v$, $d$ is
  uniformly distributed in~$\X_4$. Since this holds for any $\tuple v$, $d$ is
  uniformly distributed and independent of~$\z$ as desired and $\psi(\x,\y)$
  is satisfied by~$\X$.

  Suppose now that a team $\X$ satisfies $\psi(\x,\y)$. We want to show
  that $\x\approx \y$.  But the chain of reasoning above also works
  ``backwards''. Fix a value $\tuple v$ of $\x$. We want to show that
  $|\X_{\x=\tuple v}|=|\X_{\y=\tuple v}|$.  It is clear that it is sufficient to look
  at $\X_3$ as defined above.  But because $\theta$ says that $d$ is
  a ``detector'' of whether $\z=\x$ or not, it is in fact sufficient
  to check $\x\approx \y$ for the subteam $\X_4$ (also as defined
  above).  But in $\X_4$, this follows from~$\f$.\qed
\end{proof}

\section{Proof of Theorem \ref{thm:marg}}\label{a:thm}

Theorem \ref{thm:marg} follows from Lemma \ref{apulemma} presented below. Lemma \ref{apulemma} can be proven following the proof of Theorem 2 in \cite{HKMV18}. We omit the details and instead delineate intuition behind the translation. The idea is to simulate the semantics of the probabilistic conditional independence atom using only marginal independence and marginal identity atoms. First, the universally quantified $\tuple y$ in the translation represents all possible variable assignments $s$ of $\tuple x$. Second, $\psi_0$ and $\psi_1$ indicate that the marginal distributions of $\tuple x_0$, $\tuple x_0\tuple x_1$, $\tuple x_0\tuple x_2$, and $\tuple x_0\tuple x_1\tuple x_2$ are distributed respectively to $\tuple z_0,\tuple z_1,\tuple z_2,\tuple z_3$ independently of $\tuple y$ and of each other. Third, $\psi_2$ encodes the product of the weights of $s(\tuple x_0)$ and $s(\tuple x_0\tuple x_1\tuple x_2)$ by $\a=0$, and $\psi_3$ similarly the product of the weights  of $s(\tuple x_0\tuple x_1)$ and $s(\tuple x_0\tuple x_2)$ by $\b=0$. Finally, conditional independence between $\tuple x_1$ and $\tuple x_2$ given $\tuple x_0$ follows iff these products are equal relative to all assignments of $\tuple y$. Theorem \ref{thm:marg} then follows from this lemma since the constant $0$ and the marginal identity atom are both definable in $\FO(\pind)$.
\begin{lemma}\label{apulemma}
Let $\tuple x_0, \tuple x_1, \tuple x_2$ be three sequences of variables from $\tuple x=(x_1, \ldots ,x_n)$, and let $0$ be a constant symbol. Then $\pci{\tuple x_0}{\tuple x_1}{\tuple x_2}$ is equivalent to 
\[\phi:=\forall\tuple y\exists \tuple z_0\tuple z_1\tuple z_2 \tuple z_3\a\b(\psi_0\wedge \psi_1\wedge \psi_2\wedge \psi_3\wedge \psi_4)\]
where
\begin{align*}
&\psi_0:=\pmi{\tuple y}{\tuple  z_0}\wedge \pmi{\tuple y\tuple  z_0}{\tuple  z_1}\wedge \pmi{\tuple y\tuple  z_0\tuple  z_1}{\tuple  z_2}\wedge \pmi{\tuple y\tuple  z_0\tuple  z_1\tuple  z_2}{\tuple  z_3},\\
&\psi_1:=\tuple x_0 \approx \tuple z_0 \wedge \tuple x_0\tuple x_1 \approx \tuple z_1 \wedge \tuple x_0\tuple x_2 \approx \tuple z_2 \wedge \tuple x_0 \tuple x_1\tuple x_2\approx \tuple z_3, \\
&\psi_2:= \a=0\leftrightarrow (\tuple z_0=\tuple y_0 \wedge \tuple z_3 = \tuple y_0\tuple y_1 \tuple y_2),\\
&\psi_3:= \b =0\leftrightarrow (\tuple z_1=\tuple y_0\tuple y_1 \wedge \tuple z_2 = \tuple y_0 \tuple y_2),\\
&\psi_4:= \tuple y\a \approx \tuple y \b.
\end{align*}
\end{lemma}

\section{Proof of Proposition  \ref{prop:uc}}\label{a:prop2}

\repeatproposition{prop:ucl}
\begin{proof}
  We may assume that $\X=(X,f)$ and $\Y=(X,g)$. We prove the claim by
  structural induction on $\phi$. We omit the cases for atomic
  formulae and conjunction which are straightforward.
  \begin{itemize}
  \item Assume that $\phi = \phi_0\vee \phi_1$. By the semantics of
    the disjunction, we find $p,q\in[0,1]$ and distributions
    $f_0,f_1,g_0,g_1$ over $X$ such that $\A\models_{(X,f_0)} \phi_0$,
    $\A\models_{(X,f_1)} \phi_1$, $\A\models_{(X,g_0)} \phi_0$,
    $\A\models_{(X,g_1)} \phi_1$, $f= pf_0 + (1-p)f_1$, and
    $ g= qg_0 + (1-q)g_1$. Define
    $h_0:=\frac{kpf_0+ (1-k)qg_0}{kp+(1-k)q}$ and
    $h_1:=\frac{k(1-p)f_1+ (1-k)(1-q)g_1}{k(1-p)+ (1-k)(1-q)}$. By the
    induction hypothesis $\A\models_{(X,h_0)} \phi_0$ and
    $\A\models_{(X,h_1)} \phi_1$, since
    $(X,h_0)= (X,f_0)\qcup{a} (X,g_0)$ for $a:= \frac{kp}{kp+(1-k)q}$,
    and $(X,h_1)= (X,f_1)\qcup{b} (X,g_1)$ for
    $b:= \frac{k(1-p)}{k(1-p)+ (1-k)(1-q)}$.  Then
    $(X,f)\kcup (X,g)=(X,h_0)\qcup{c} (X,h_1)$ for $c:=kp+ (1-k)q$
    because
    \begin{align*}
      ch_0+(1-c)h_1&= c\frac{kpf_0+ (1-k)qg_0}{c} + (1-c)\frac{k(1-p)f_1+ (1-k)(1-q)g_1}{1-c}\\
                   &=k[pf_0 + (1-p)f_1] + (1-k)[qg_0 + (1-q)g_1]\\ 
                   &= kf + (1-k)g.
    \end{align*} 
    Consequently, $\A \models_{(X,f)\kcup (X,g)} \phi_0 \vee\phi_1$ follows from the semantics of the disjuction which completes the disjunction step of the induction.
  \item Assume that $\phi = \forall x \psi$. Then
    $\A \models_{\X[A/x]} \psi$ and $\A \models_{\Y[A/x]} \psi$, and
    by induction assumption $\A \models_{\X[A/x] \kcup \Y[A/x]} \psi$.
    The claim then follows since
    $\X[A/x] \kcup \Y[A/x]= (\X \kcup \Y)[A/x]$.

  \item Assume that $\phi = \exists x \psi$. Then
    $\A \models_{\X[F/x]} \psi$ and $\A \models_{\Y[G/x]} \psi$ where
    $F$ and $G$ are functions that map each $s\in X$ to a probability
    distribution $F_s$ over $A=\Dom(\A)$. We let $H$ be a function
    that maps $s \in X$ to a probability distribution $H_s$ over $A$
    such that
    \[H_s(a):= \frac{kf(s)F_s(a)+(1-k)g(s)G_s(a)}{kf(s)+(1-k)g(s)}.\]
    Note that $\sum_{a\in A} H_s(a) = 1$ follows from
    $\sum_{a\in A} F_s(a)=\sum_{a\in A} G_s(a)= 1$. By induction
    assumption $\A \models_{\X[F/x]\kcup \Y[F/x]} \psi$. The claim now
    follows from $\X[F/x]\kcup \Y[F/x] = (\X\kcup \Y)[H/x]$, which
    holds since for all $a\in A$:
    \[kf(s)F_s(a) + (1-k)g(s)G_s(a) = [kf(s) + (1-k)g(s)]H_s(a).\]
    This concludes the case of existential quantification and the
    proof.\qed
  \end{itemize}
\end{proof}

\end{document}